\documentclass{article}

\usepackage{amsthm}
\usepackage{amsmath}
\usepackage{graphicx}
\usepackage{bbm}
\usepackage{amsfonts}

\newcommand{\beq}{\begin{eqnarray}}
\newcommand{\eeq}{\end{eqnarray}}
\newcommand{\la}{\langle}
\newcommand{\ra}{\rangle}

\newtheorem{theorem}{Theorem}[section]
\newtheorem{lemma}{Lemma}[section]

\newtheorem{cor}{Corollary}[section]
\newtheorem{definition}{Definition}[section]

\theoremstyle{remark}
\newtheorem*{remark}{Remark}

\begin{document}

\title{Complete hierarchies of SIR models on arbitrary networks with exact and approximate moment closure}

\author{Kieran J. Sharkey\textsuperscript{a}\footnote{Corresponding author (Email: kjs@liv.ac.uk; Tel: +44 (0)151 794 4023).} and Robert R. Wilkinson\textsuperscript{a}}
\date{}
\maketitle
\noindent (a) Department of Mathematical Sciences, University of Liverpool, Peach Street, Liverpool, L69 7ZL

\section*{Abstract}
We first generalise ideas discussed by Kiss et al. (2015) to prove a theorem for generating exact closures (here expressing joint probabilities in terms of their constituent marginal probabilities) for susceptible-infectious-removed (SIR) dynamics on arbitrary graphs (networks). For Poisson transmission and removal processes, this enables us to obtain a systematic reduction in the number of differential equations needed for an exact `moment closure' representation of the underlying stochastic model. We define `transmission blocks' as a possible extension of the block concept in graph theory and show that the order at which the exact moment closure representation is curtailed is the size of the largest transmission block. More generally, approximate closures of the hierarchy of moment equations for these dynamics are typically defined for the first and second order yielding mean-field and pairwise models respectively. It is frequently implied that, in principle, closed models can be written down at arbitrary order if only we had the time and patience to do this. However, for epidemic dynamics on networks, these higher-order models have not been defined explicitly. Here we unambiguously define hierarchies of approximate closed models that can utilise subsystem states of any order, and show how well-known models are special cases of these hierarchies.
\section{Introduction}
\label{1}
A primary method for incorporating spatial structure and other contact structures into epidemic models is to use a network of contacts \cite{Newman03}. While simulation of stochastic models is straightforward on these networks, obtaining differential equation descriptions of the relevant time series is more complex. Here we consider the construction of a hierarchy of moment equations which, in statistical physics, is sometimes known as the Bogoliubov-Born-Green-Kirkwood-Yvon (BBGKY) hierarchy after the names of its originators. The method was applied to population-level network-based epidemic and ecological models in the 1990s where truncation of the hierarchy was made at second order yielding pair-approximation models \cite{Matsuda, Sato, Harada, Keeling99}. Higher-order truncation of this hierarchy at the level of triples has also been investigated  \cite{Matsuda, Bauch, House}. With increasing computational resources it has also become numerically viable to consider these hierarchies in terms of individuals, rather than population-level quantities \cite{Sharkey08, Sharkey11, Simpson11}. A particularly important feature of the individual-level representation is that it enables us to establish exactness for finite populations in certain circumstances (see \cite{Sharkey13} and also \cite{Kiss}  and \cite{Wilkinson} by different methods).

Here we generalise ideas discussed by Kiss et al. \cite{Kiss} and also noted by Newman \cite{Newman10}, and apply them to arbitrary directed networks. We also observe that they apply to non-Markovian as well as Markovian SIR dynamics. Depending on the network, we find that for Markovian dynamics, exact closed models exist at all levels of the hierarchy of moment equations. The exact models and exact closures considered in \cite{Sharkey13} and \cite{Kiss} then represent special cases.

While the majority of moment closure models do not go beyond closure at the level of pairs (second order), it is frequently stated that, in principle, closed models at any order can be constructed. However, such higher-order models are rarely defined explicitly. Here, in the Markovian SIR epidemic context, we shall define hierarchies of closed models that can be constructed unambiguously at all orders by a systematic truncation method. In fact, we shall define and investigate several hierarchies of approximate models. All of these converge to exact representations at truncation orders which depend on the underlying network structure and all of them have either the pair-level model discussed in~\cite{Sharkey08} and~\cite{Sharkey11} or the variant of this model discussed in~\cite{Sharkey13} as the lowest (zeroth) order level of truncation.

The next section discusses the relevant background concepts upon which our work builds. Section~\ref{3} introduces the exact closure theorem which defines the conditions under which simplifications to the hierarchy of equations can be made for particular networks. Section~\ref{4} introduces approximate closures leading to complete hierarchies of approximate models.

\section{Background concepts}
\label{2}

Apart from Theorem~\ref{general_partitioning_theorem} which applies more generally, we shall consider a Markovian class of SIR models on contact networks. In particular, we consider a directed graph $D=(V,A)$ consisting of $N=|V|$ individuals/nodes and a set $A$ of arcs. We also label each individual according to some arbitrary ordering such that if $i \in V$ then $i \in \{1, 2, \ldots, N  \}$. Each individual can be in only one of three states $S$, $I$, or $R$ at any given time. Node $j \in V$, when infectious, makes `infectious contacts' to node $i \in V$ via a Poisson process of rate $T_{ij} \ge 0$, where $T_{ij}>0 \Leftrightarrow (j,i) \in A$ and where we assume that $T_{ii}=0$ for all $i\in V$. If node $i$ is susceptible when it receives an infectious contact then it immediately becomes infectious. It will then remain infectious for an exponentially distributed period, with parameter $\gamma_i$, after which it becomes recovered which is an absorbing state for the individual. We thus have a continuous-time Markov chain with a state space of size $3^N$. Except where otherwise stated, we also assume initial conditions such that the states of all nodes are initially statistically independent. This assumption encompasses all pure-state initial conditions, such as a specific individual being infectious with all others susceptible, and it also incorporates binomially distributed initial conditions. Uniform initial conditions can also be exactly represented with additional computation \cite{Wilkinson}.

\begin{definition}
$S_i$, $I_i$ and $R_i$ denote the indicator random variables for the events that node $i \in V$ is susceptible, infectious and removed respectively. Depending on the context, it will also be convenient to refer to $S_i$, $I_i$ and $R_i$ as the corresponding events themselves. 
\label{S_I_indicators}
\end{definition}

The hierarchy comprises of a sequence of equations containing the first moments and mixed moments of the random variables $S_i$ and $I_i$. Using angle brackets to denote expectation values, it can be shown \cite{Sharkey11} that the master equation (or Kolmogorov forward equations) implies the following rate equations:
\beq \nonumber
\dot{\la S_i \ra}&=&- \sum_{j=1}^{N}T_{ij}\la S_i I_j \ra, \\
\dot{\la I_i \ra}&=& \sum_{j=1}^{N}T_{ij}\la S_i I_j \ra - \gamma_i \la I_i \ra.
\label{singlet}
\eeq
where $S_iI_j$ is a product of the indicator random variables which also specifies a state of the subsystem of order two comprising of the pair of nodes $i$ and $j$. For this pair state we have:
\beq \nonumber 
\dot{\la S_i I_j \ra}&=& \sum_{k \neq i}^N T_{jk}\la S_i S_j I_k \ra - \sum_{k \neq j}^{N}T_{ik}\la S_i I_jI_k  \ra \\
&&-T_{ij} \la S_i I_j \ra - \gamma_j \la S_i I_j \ra
\label{doublet}.
\eeq

More generally, for these models, the master equation allows us to write down a rate equation for the probability of any subsystem state of size $n$ in terms of subsystem states of size $n$ and subsystem states of size $n+1$. We state this as a theorem below (Theorem~\ref{theorem1}).

Following prior work \cite{Sharkey13}, but with a notational simplification brought about by using the same index for all system and subsystem states, we define an alternative notation to Definition~\ref{S_I_indicators} that is useful for keeping track of the hierarchy of moment equations in this context. 

\begin{definition}
We use the following notation to denote subsystem states.
\begin{itemize}
\item $\psi_W$ is a subsystem comprising of the set of nodes $W \subset V$.
 \item Let $A$ be a mapping from the elements of $W$ to $\{S, I, R\}$, and let $A_i$ be the image of node $i \in W$ under $A$. Thus, $A$ can also be interpreted as a pure state for subsystem $\psi_W$, i.e. the state where, for all $i\in W$, individual $i$ is in state $A_i$. 
\item $\psi_W^A$ denotes the indicator random variable for the event that $\psi_W$ is in state $A$. Thus the probability of the event that subsystem $\psi_W$ is in state A is $P(\psi_W=A)=\la\psi_W^A\ra$. As in Definition~\ref{S_I_indicators}, it is also convenient to refer to $\psi_W^A$ as the event that $\psi_W$ is in state $A$. 
\end{itemize}
\label{indicator}
\end{definition}
\begin{remark}
For the event where node $i$ is in a susceptible state, we can draw the following correspondence between the notations: $\psi_i^S=S_i$, and similarly for the infectious and removed states.
\end{remark}
\begin{definition}
\label{h_def}
Let $k\in W\subseteq V$ and $X \in \{S,I,R\}$ and let $A$ be a state of subsystem $\psi_W$. Then, $h_k^X(\psi^A_W)$ denotes the indicator random variable or event $\psi^A_W$, but where the state of node $k$ is changed to state $X$ leaving the states of all other nodes unchanged. Note that if $A_k=X$ then $h_k^X(\psi^A_W)=\psi^A_W$.
\end{definition}

\begin{theorem}
For any subsystem $\psi_W$, the probability that it is in state $A$ is governed by the rate equation:
\begin{eqnarray}\nonumber
\dot{\la \psi^A_W\ra}&=&\sum_{k \in W} \mathbbm{1}(A_k= S) \Bigg[ - \sum_{n \in W} T_{kn} \mathbbm{1}(A_n= I)  \la \psi^A_W \ra - \sum_{n \in V \setminus W} T_{kn} \la  \psi^A_W I_n \ra  \Bigg] \\ \nonumber
&&+ \sum_{k \in W} \mathbbm{1}(A_k= I) \Bigg[  \sum_{n \in W} T_{kn} \mathbbm{1}(A_n= I) \la h_k^S( \psi^A_W) \ra - \gamma_k \la \psi^A_W \ra \Bigg]\\ \nonumber
&& +\sum_{k \in W} \mathbbm{1}(A_k= I) \Bigg[  \sum_{n \in V \setminus W} T_{kn} \la h_k^S ( \psi^A_W) I_n \ra
 \Bigg] \\ 
&& + \sum_{k \in W} \mathbbm{1}(A_k= R) \Bigg[  \gamma_k \la h_k^I (\psi^A_W) \ra \Bigg],
\label{14}
\end{eqnarray}
where here, and throughout this paper, the indicator $\mathbbm{1}(.)$ is equal to 1 if its argument is true and is equal to zero otherwise.
\label{theorem1}
\end{theorem}

This theorem is proved in \cite{Sharkey13}. Starting with subsystem states that are only composed of susceptible or infectious individuals, repeated application of equation~\ref{14} to each of these states as well as to any subsystem states that arise on its right-hand side can never result in subsystem states with a removed individual. This is due to the absence of $h_k^R$ in equation \ref{14}. Hence, for these subsystem states, $\mathbbm{1}(A_k= R)=0 $ for all $k\in W$ so equation~\ref{14} becomes: 
\begin{eqnarray}\nonumber
\dot{\la \psi^A_W\ra}&=&\sum_{k \in W} \mathbbm{1}(A_k= S) \Bigg[ - \sum_{n \in W} T_{kn} \mathbbm{1}(A_n= I)  \la \psi^A_W \ra - \sum_{n \in V \setminus W} T_{kn} \la  \psi^A_W I_n \ra  \Bigg] \\ \nonumber
&&+ \sum_{k \in W} \mathbbm{1}(A_k= I) \Bigg[  \sum_{n \in W} T_{kn} \mathbbm{1}(A_n= I) \la h_k^S( \psi^A_W) \ra - \gamma_k \la \psi^A_W \ra \Bigg]\\ 
&& +\sum_{k \in W} \mathbbm{1}(A_k= I) \Bigg[  \sum_{n \in V \setminus W} T_{kn} \la h_k^S ( \psi^A_W) I_n \ra
 \Bigg].
\label{subsystem}
\end{eqnarray}
Equations~\ref{singlet} and~\ref{doublet} can now be seen to be special cases of this theorem. 

By applying equation \ref{subsystem} to every individual in the network for states $S$ and $I$ and then reapplying to every new subsystem state which emerges, we obtain a closed set of differential equations for a set $M$ of subsystem states. However, $|M|$ will generally be very large for most systems, preventing numerical solution. 

To reduce the number of equations, we need to introduce a mechanism to curtail the generation of new subsystem states. In the next section, we discuss scenarios in which this can be done where the emerging system is still an exact representation of the underlying stochastic process. Following this, we shall consider hierarchies of approximate closed models. 

\section{Exact closed models}
\label{3}
Here we prove a theorem pertaining to arbitrary SIR dynamics on arbitrary networks. We then use this to derive a class of exact models for Markovian SIR dynamics on arbitrary networks. We illustrate this with some examples, and finally state a theorem specifying the maximum size of subsystem needed to exactly represent the dynamics on any given network.
\subsection{Exact closure theorem}
\label{3.1}
For a given directed graph $D=(V,A)$ with set $V$ of nodes/individuals and set $A$ of arcs, we make the following definitions:

\begin{definition}
$\mbox{IN}(X)$ is the set of individuals that can reach at least one member of $X\subseteq V$ by following a directed path. Note that $X\subseteq \mbox{IN}(X)$.
\end{definition}
\begin{definition}
Let $X,Y,Z \subset V$ be disjoint and non-empty. The set of nodes $Z$ is `dynamically partitioning' with respect to $X$ and $Y$ if and only if we have $f_E (X,Y,Z)=1$ where:
\beq f_E (X,Y,Z)= \begin{cases}
1 & \mbox{if } \mbox{IN}(X)\cap \mbox{IN}(Y) = \emptyset \quad \mbox{\rm{(}in } D - Z)     \\
0 & \mbox{otherwise}
\end{cases}
\eeq
and $D-Z$ is the vertex-set deleted subgraph consisting of nodes $V \setminus Z$. Here, $E$ is chosen to represent `exact'; this is appropriate since we shall now see that $f_E (X,Y,Z)=1$ implies the existence of an exact closure relation.
\end{definition}
\begin{remark}
If the network is undirected then $f_E (X,Y,Z)=1$ if and only if there is no path between $X$ and $Y$ in $D-Z$. 
\end{remark}
\begin{theorem}
We consider stochastic SIR dynamics defined on a time-invariant network where the initial conditions are such that the states of individual nodes are initially statistically independent. Let $\psi^A_X, \psi^B_Y$ and $\psi^C_Z$ be indicator random variables or events where $X,Y,Z \subseteq V$ are disjoint and nonempty. If $Z$ is dynamically partitioning with respect to $X$ and $Y$, and all nodes in subsystem state $C$ are susceptible $(C_i=S\;\;\forall i\in Z)$, then provided that $\la \psi^C_Z \ra\neq 0$,
\beq
\la \psi^A_X\psi^B_Y\psi^C_Z \ra  =  \frac{ \la \psi^A_X\psi^C_Z \ra \la \psi^B_Y\psi^C_Z \ra}{  \la \psi^C_Z \ra}. 
\eeq
\label{general_partitioning_theorem}
\end{theorem}
\begin{proof}
If the infection has not passed through $Z$ (which is guaranteed by all nodes in state $C$ being susceptible), the states of the individuals in $X$ are statistically independent of the states of the individuals in $Y$. This is true since $f_E(X,Y,Z)=1$ implies that there are no individuals from which both a member of $X$ and a member of $Y$ can be reached without traversing a member of $Z$. We have:
\beq
P(\psi^A_X, \psi^B_Y, \psi^C_Z|\psi^C_Z)=P(\psi^A_X, \psi^C_Z |\psi^C_Z) P(\psi^B_Y,\psi^C_Z|\psi^C_Z).
\nonumber
\eeq
Given that $P(\psi^C_Z)\neq 0$, we have:
\beq
\frac{P(\psi^A_X, \psi^B_Y, \psi^C_Z)}{P(\psi^C_Z)}=\frac{P(\psi^A_X, \psi^C_Z)}{P(\psi^C_Z)}\frac{P(\psi^B_Y,\psi^C_Z)}{P(\psi^C_Z)},
\nonumber
\eeq
from which the result follows.
\end{proof}
\begin{remark}
For the case of zero denominator, note that $P(\psi^C_Z)= 0$ implies that $P(\psi^A_X, \psi^B_Y, \psi^C_Z)=0$. 
\end{remark}
Notice that we made no assumptions about the SIR dynamics in proving this theorem and that it is therefore not restricted to Markovian systems, although it is the Markovian case that we shall be applying it to in the remainder of this paper.

The theorem is a generalisation of the main result in \cite{Kiss} which is stated in terms of single dynamically partitioning individuals on undirected networks. In that context they are referred to simply as partitioning individuals due to their correspondence  to graph partitioning. Some examples of where the exact closure theorem can be applied are shown in Figure~\ref{general_exact_closures}.
\begin{figure}
   \centerline{\includegraphics[width=.8\textwidth]{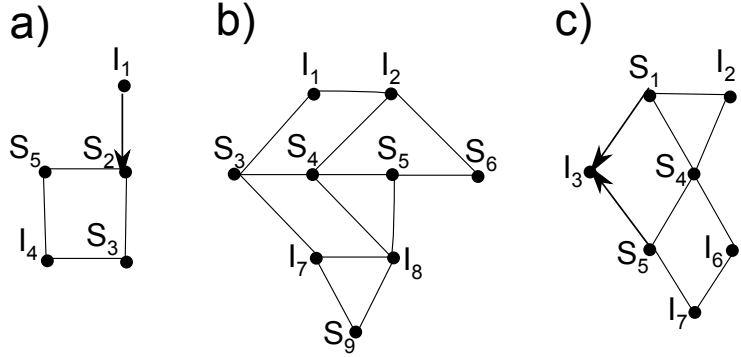}}
    \caption{Three examples of network states where the location of the susceptible nodes allows the application of the exact closure theorem. Here directed links have arrowheads and undirected links do not.}
    \label{general_exact_closures}
\end{figure}
In this Figure and throughout the remainder of the paper, network links without arrowheads denote undirected links whereas those with arrowheads denote directed links. Figure~\ref{general_exact_closures}a is typical of the dynamical partitioning we shall consider in this paper. Applying Theorem~\ref{general_partitioning_theorem}, we see that there is dynamical partitioning about node 2, so we have $\la I_1S_2S_3I_4S_5\ra=\la I_1S_2\ra\la S_2S_3I_4S_5\ra/\la S_2\ra$. For Figure~\ref{general_exact_closures}b we can dynamically partition about a cluster of susceptible nodes. In fact there are two exact closures we can write down: $\la I_1I_2S_3S_4S_5S_6I_7I_8S_9\ra=\la I_1I_2S_3S_4S_6\ra\la S_3S_4S_5S_6I_7I_8S_9\ra/\la S_3S_4S_6\ra=\la I_1I_2S_3S_4S_5S_6\ra\la S_3S_4S_5I_7I_8S_9\ra/\la S_3S_4S_5\ra$. In Figure~\ref{general_exact_closures}c we can apply the exact closure theorem to obtain $\la S_1I_2S_4S_5I_6I_7\ra=\la S_1I_2S_4\ra\la S_4S_5I_6I_7\ra/\la S_4\ra$. Note that $I_3$ is not included in this closure.

For our purposes, we are interested in a special case of the exact closure theorem which is captured by the following corollary.
\begin{cor}
For subsystem state $A$ of $\psi_W$, if $A_k=S$ where $k\in W$, and if $f_E(n,W \setminus k,k)=1$ where $n\in V\setminus W$, then
\beq
\la  \psi^A_W I_n \ra  =  \frac{  \la \psi^A_W \ra  \la  S_k I_n \ra}{  \la S_k \ra}. 
\label{eqclosure}
\eeq
\label{partitioning_theorem}
\end{cor}
This corollary is illustrated by the example in Figure~\ref{general_exact_closures}a. By applying this to equation~\ref{subsystem} we obtain:
\begin{eqnarray} \nonumber
\dot{\la \psi^A_W\ra}&=&\sum_{k \in W} \mathbbm{1}(A_k= S) \Bigg[ - \sum_{n \in W} T_{kn} \mathbbm{1}(A_n= I)  \la \psi^A_W \ra \Bigg] \\ \nonumber
&& + \sum_{k \in W} \mathbbm{1}(A_k= I) \Bigg[  \sum_{n \in W} T_{kn} \mathbbm{1}(A_n= I) \la h_k^S( \psi^A_W) \ra - \gamma_k \la \psi^A_W \ra \Bigg] \\ \nonumber
&&- \sum_{k \in W} \mathbbm{1}(A_k= S) \sum_{n \in V \setminus W}T_{kn} \Bigg[  \Big( 1-f_E(n,W \setminus k,k) \Big) \la h_k^S(\psi^A_W) I_n\ra \\ \nonumber
&& + f_E(n,W \setminus k,k)\frac{\la h^S_k(\psi^A_W)\ra \la S_k I_n\ra}{\la S_k\ra}  \Bigg] \\ \nonumber
&&+ \sum_{k \in W} \mathbbm{1}(A_k= I) \sum_{n \in V \setminus W}T_{kn} \Bigg[  \Big( 1-f_E(n,W \setminus k,k) \Big) \la h_k^S(\psi^A_W) I_n\ra \\ 
&& + f_E(n,W \setminus k,k)\frac{\la h^S_k(\psi^A_W)\ra \la S_k I_n\ra}{\la S_k\ra}  \Bigg].
\label{part_eq1}
\end{eqnarray}

For an arbitrary network, by applying equation~\ref{part_eq1} to the indicator random variables $S_i$ and $I_i$ for all $i\in \{1,2,...,N\}$,  and then reapplying it to every new subsystem state that emerges, a closed set of differential equations for the exact time-evolution of the probability of an individual being in a particular state is obtained for all individuals. The number of equations that will be needed is limited by the closures that are made possible by the exact closure theorem. 
\begin{definition}
For a given network, the induced set $M_E$ of subsystem states is obtained by applying equation~\ref{part_eq1} to every individual \rm (\it for states $S$ and $I$\rm )\it in the network, and then reapplying to every new subsystem state that emerges. $M_E$ is then the full set of subsystem states that emerge during this process. 
\end{definition}
\begin{remark}
It follows that $ S_i$ and $I_i$ $(\forall i \in V)$ and $S_i I_j$ $(\forall i,j \in V : T_{ij}>0)$ represent members of $M_E$ for any network.
\end{remark}

\subsection{Examples}
Before determining the network structures under which dynamical partitioning occurs more generally, we consider some examples. For further examples in the context of undirected networks the reader is directed to \cite{Kiss}.
\subsubsection{Example 1}
\begin{definition}
A network is a tree network if and only if its underlying graph (all directed edges are replaced by undirected edges) is a tree or forest.
\end{definition}
\begin{theorem}
For Markovian SIR dynamics on a tree network where the states of all individuals are initially statistically independent, the following equations hold exactly:
\begin{eqnarray}\nonumber
\dot{\la S_i\ra}&=&-\sum_{j=1}^N T_{ij}\la S_iI_j\ra, \\ \nonumber \dot{\la I_i\ra}&=&\sum_{j=1}^N
T_{ij}\la S_iI_j\ra -\gamma_i\la I_i\ra,\\ \nonumber
\dot{\la S_iI_j\ra}&=&\sum_{k=1,k\neq i}^N T_{jk}\frac{\la S_iS_j\ra\la S_jI_k\ra}{\la S_j\ra}-\sum_{k=1,k\neq j}^N T_{ik}\frac{\la S_iI_k\ra\la S_iI_j\ra}{\la S_i\ra},    \\ \nonumber
&&-T_{ij}\la S_iI_j\ra-\gamma_j\la S_iI_j\ra, \\ 
\dot{\la S_iS_j\ra}&=&-\sum_{k=1,k\neq j}^NT_{ik}\frac{\la S_iS_j\ra\la S_iI_k\ra}{\la
S_i\ra}-\sum_{k=1,k\neq i}^NT_{jk}\frac{\la S_iS_j\ra\la S_jI_k\ra}{\la S_j\ra}.
\label{pair_level}
\end{eqnarray}
\end{theorem}
\begin{proof}
For such tree networks, every individual is dynamically partitioning relative to any two of its neighbours on the underlying graph. Hence, the above system follows directly from repeated application of equation~\ref{part_eq1}, starting with $\la S_i\ra$, $\la I_i\ra$ $\forall i \in V$.
\end{proof}
\begin{remark}
This is the pairwise model that was shown to be exact for tree networks in \cite{Sharkey13}.
\end{remark}

\subsubsection{Example 2}
 \label{3.2}
Consider the graph in Figure~\ref{partitioning_examples}a. 
\begin{figure}
   \centerline{\includegraphics[width=.8\textwidth]{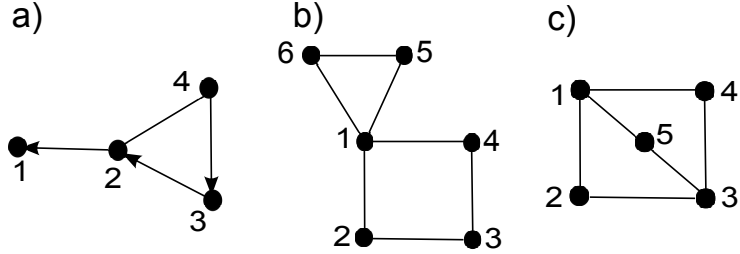}}
    \caption{Some example graphs. For dynamics on these graphs, we assume a generic removal rate $g$ and a transmission rate of 1 across all links.}
    \label{partitioning_examples}
\end{figure}
Let us suppose that all nodes have the same removal rate $g$ and that the transmission rate across all network links is unity. For simplicity we shall also make this assumption through the remainder of the explicit examples in this paper. We can apply Corollary~\ref{partitioning_theorem} which is embedded in equation~\ref{part_eq1} to build up the induced subsystem states $M_E$. Let us just consider the infectious probability of node 1 to see how this works. We have:
\beq\nonumber
\dot{\la I_1\ra}=\la S_1I_2\ra-g\la I_1\ra.
\eeq
Here and throughout the paper we order nodes according to the numerical order of their labels; the relevant motif structures need to be understood with reference to the associated graph. Now, node 2 is dynamically partitioning with respect to nodes 1 and 3, and it is also dynamically partitioning with respect to nodes 1 and 4. Hence:
\beq \nonumber
\dot{\la S_1I_2\ra}&=&\la S_1S_2I_3\ra+\la S_1S_2I_4\ra-\la S_1I_2\ra-g\la S_1I_2\ra \\ \nonumber
&=&\frac{\la S_1S_2\ra\la S_2I_3\ra}{\la S_2\ra}+\frac{\la S_1S_2\ra\la S_2I_4\ra}{\la S_2\ra}-(1+g)\la S_1I_2\ra.
\eeq
Rather than a complete analysis of all induced subsystem states that arise, we take the single pair state $S_2I_3$ from this equation as an example. Here, node 3 is not dynamically partitioning with respect to nodes 2 and 4 but node 2 is dynamically partitioning with respect to 1 and 3 so:
\beq\nonumber
\dot{\la S_2I_3\ra}=\la S_2S_3I_4\ra-\frac{\la I_1S_2\ra\la S_2I_3\ra}{\la S_2\ra}-\la I_4S_2I_3\ra-(1+g)\la S_2I_3\ra.
\eeq
Then for $\la S_2S_3I_4\ra$, node 2 is dynamically partitioning with respect to node 1 and nodes 3 and 4 so:
\beq\nonumber
\dot{\la S_2S_3I_4\ra }=-\frac{\la I_1S_2\ra\la S_2S_3I_4\ra}{\la S_2\ra}-(2+g)\la S_2S_3I_4\ra.
\eeq
We see that here, $M_E$ represents a significant dimensional reduction in the number of induced subsystem states compared to the full set of induced subsystem states $M$.
\subsubsection{Example 3}
For the undirected graph in Figure~\ref{partitioning_examples}b there is dynamical partitioning about node 1. Starting with (for example) the infectious probability for node 1, we have:
\beq\nonumber
\dot{\la I_1\ra}=\la S_1I_2\ra+\la S_1I_4\ra+\la S_1I_5\ra + \la S_1I_6\ra-g\la I_1\ra,
\eeq
where again we are assuming transmission rates of unity and a removal rate $g$ for each node. Now, choosing the first of these pairs to develop one part of the induced set $M_E$ gives:
\beq \nonumber
\dot{\la S_1I_2\ra}&=&\la S_1S_2I_3\ra-\la S_1I_2I_4\ra-\frac{\la S_1I_2\ra\la S_1I_5\ra}{\la S_1\ra}-\frac{\la S_1I_2\ra\la S_1I_6\ra}{\la S_1\ra}\\
&&-(1+g)\la S_1I_2\ra,
\label{squarepairs}
\eeq
and then for the first of these triples:
\begin{eqnarray}\nonumber
\dot{\la S_1S_2I_3\ra}&=&\la S_1S_2S_3I_4\ra-\la S_1S_2I_3I_4\ra-\frac{\la S_1S_2I_3\ra\la S_1I_5\ra}{\la S_1\ra}-\frac{\la S_1S_2I_3\ra\la S_1I_6\ra}{\la S_1\ra} \\ \nonumber 
&&-(1+g)\la S_1S_2I_3\ra.
\label{eqp1}
\end{eqnarray}
For the first of these quads we have:
\beq\nonumber
\dot{\la S_1S_2S_3I_4\ra}&=&-\frac{\la S_1S_2S_3I_4\ra\la S_1I_5\ra}{\la S_1\ra}-\frac{\la S_1S_2S_3I_4\ra\la S_1I_6\ra}{\la S_1\ra} \\ \nonumber
&& -(2+g)\la S_1S_2S_3I_4\ra.
\eeq
Here, the maximum size of a subsystem state is four. We note that this is equal to the size of the largest simple cycle and that this was also true for example 2. However, this is not always the case as shown by the next example. 
\subsubsection{Example 4}

Figure~\ref{partitioning_examples}c shows a network where the maximum simple cycle size is 4 but the maximum size of a subsystem state in $M_E$ is 5. Starting with the infectious probability of node 1 we have:
\beq\nonumber
\dot{\la I_1\ra}=\la S_1I_2\ra+\la S_1I_4\ra + \la S_1I_5\ra -g\la I_1\ra.
\eeq
Then, taking just the subsystem state in the first term:
\beq\nonumber
\dot{\la S_1I_2\ra}=\la S_1S_2I_3\ra-\la S_1I_2I_4\ra-\la S_1I_2I_5\ra-(1+g)\la S_1I_2\ra,
\eeq
and again taking just the first term:
\beq \nonumber
\dot{\la S_1S_2I_3\ra}&=&\la S_1S_2S_3I_4\ra+\la S_1S_2S_3I_5\ra -\la S_1S_2I_3I_4\ra-\la S_1S_2I_3I_5\ra \\
&&-(1+g)\la S_1S_2I_3\ra.
\label{eqp2}
\eeq
Finally, taking the first term again gives:
\beq
\dot{\la S_1S_2S_3I_4\ra}=-2\la S_1S_2S_3I_4I_5\ra-(2+g)\la S_1S_2S_3I_4\ra .
\label{eqp3}
\eeq
In this case we see that the maximum size of a subsystem state is at the size of the system (5 nodes) and is not constrained by the largest simple cycle (4 nodes). This leads to the question: What aspect of a network specifies the largest subsystem size that appears in $M_E$? We answer this question in the following subsection.

\subsection{System size}
 \label{3.3}
Here we define the type of network structures that are amenable to dynamical partitioning. We start from single node subsystems and expand out, via equation~\ref{part_eq1}, until the largest subsystem is reached incorporating that individual before dynamical partitioning prevents larger subsystems emerging. For the undirected case, the situation simplifies considerably \cite{Kiss} since all dynamically partitioning individuals are also cut-vertices (individuals which, when removed, increase the number of connected components). It is then helpful to represent the network as a collection of blocks (maximal biconnected subgraphs) where the between-block structure is tree-like (see Figure~\ref{undirected}a). This makes it straightforward to assess the feasibility of constructing a solvable exact system by making use of dynamical partitioning.  Notice that it is possible for a node to belong to more than one block as in the top right of Figure~\ref{undirected}a although the overlap between any two blocks can only be a single node.
\begin{figure}
  \centerline{\includegraphics[width=.8\textwidth]{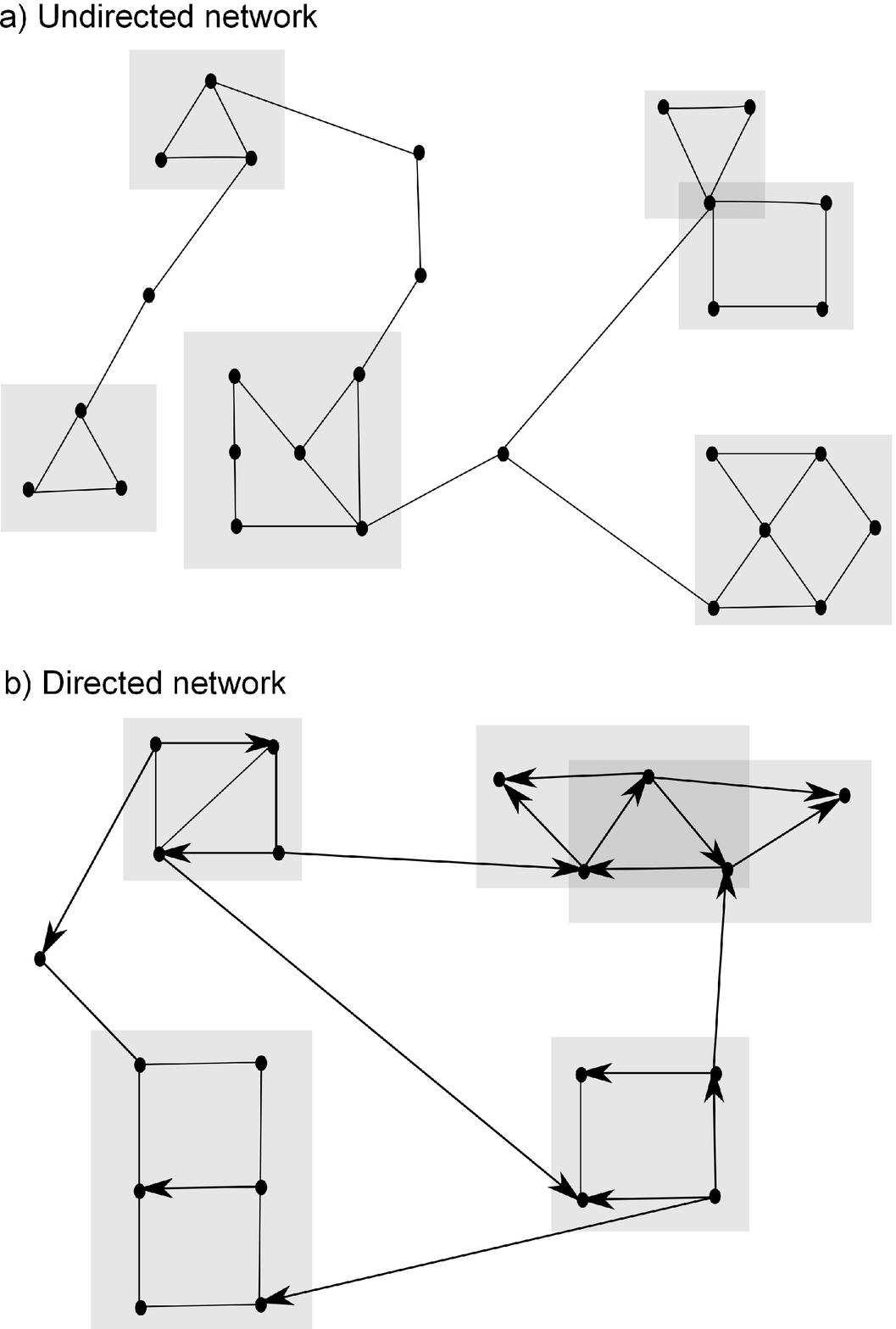}}
    \caption{Examples of networks that decompose into transmission blocks. The transmission blocks are highlighted by the shaded rectangles. Darker areas are where two transmission blocks overlap. a) An undirected network where the effectiveness of dynamical partitioning is made clear by the number of distinct biconnected subgraphs which resemble structured households. b) A directed network where identifying the transmission blocks is more complicated.}
    \label{undirected}
\end{figure}
It is interesting that this representation of the network resembles the household model structure where analytic progress can also be made \cite{Ball}. For directed networks, the situation is more complicated. Here we define `transmission blocks' to play a similar role to blocks. Indeed, blocks and transmission blocks will have equivalent definitions in the undirected case. We use the term transmission block rather than block since there are likely to be other useful extensions of the block concept for directed networks.

\begin{definition}
Let $D=(V,A)$ be a directed graph with set $V$ of nodes and set $A$ of arcs. Let $W\subseteq V$. Then $D[W]$ is the subgraph formed from the nodes of $W$ and arcs with endpoints both in $W$.
\end{definition}
\begin{definition}
The subgraph $D[W]$ is a `directed sub-block' if and only if there is at least one node reachable from all others in $D[W]$ and its underlying graph is biconnected. 
\end{definition}
\begin{remark}
According to this definition, any block in an undirected network is also a directed sub-block. Hence, the blocks illustrated in Figure~\ref{undirected}a are all directed sub-blocks.
\end{remark}
\begin{definition}
We will refer to a directed sub-block $D[W]$ as a `transmission block' if and only if there does not exist $U \supset W$ such that $D[U]$ is also a directed sub-block. 
\end{definition}

The shaded boxes in Figure~\ref{undirected} are examples of transmission blocks. Figure~\ref{undirected}b gives an example of these on a directed graph. Notice that now it is possible for transmission blocks to overlap by more than one node (the darker shaded triangle belongs to two transmission blocks). This happens when a region of the network has paths to two or more other regions that do not have paths between each other. Figure~\ref{sub-block} shows some more examples of these definitions for directed networks. Figure~\ref{sub-block}a and~b have underlying graphs that are biconnected. Figure~\ref{sub-block}{b} also has a node (node 1) which is reachable from all others and so it is a directed sub-block whereas Figure~\ref{sub-block}a is not. Figure~\ref{sub-block}b is also a transmission block since it is maximal. Additionally, neither have sub-graphs of the underlying graphs that are biconnected and so neither contain directed sub-blocks as subgraphs. Figure~\ref{sub-block}c is a transmission block (the underlying graph is biconnected and node 2 is reachable from all others). It also contains several directed sub-blocks (for example nodes 1,2 and 3). Figure~\ref{sub-block}d contains a transmission block as a subgraph (nodes 1,2,3,4) and contains several directed sub-blocks.
\begin{figure}
   \centerline{\includegraphics[width=.5\textwidth]{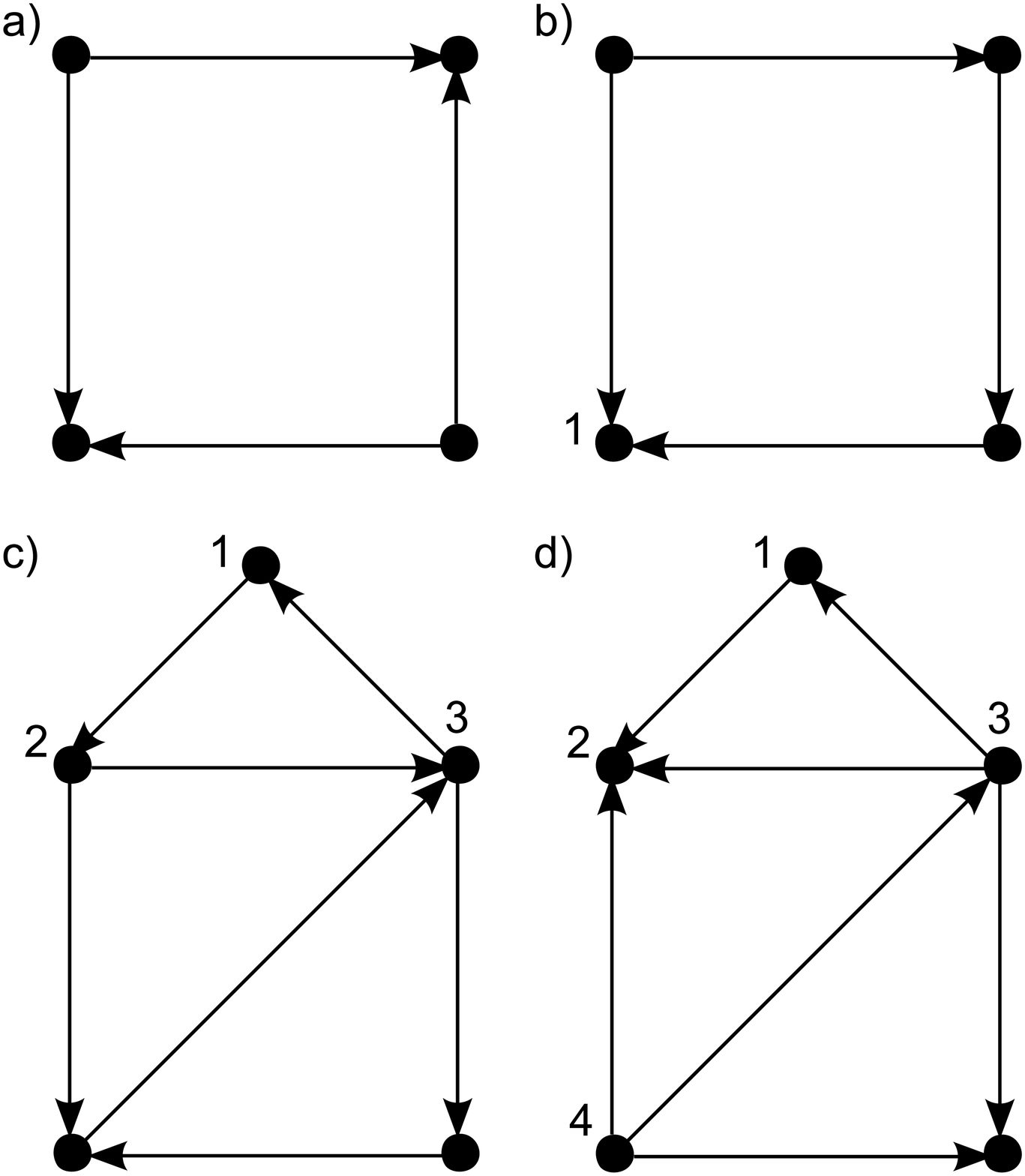}}
  \caption{Four directed graphs. Graph a) is not a transmission block whereas graphs b) and c) are transmission blocks. Graph d) contains a transmission block as a subgraph.}
 \label{sub-block}
\end{figure}

We can now state the main result on subsystem size:
\begin{theorem}
The largest subsystem state in $M_E$ consists of the same number of individuals as the largest transmission block, or it contains 2 individuals if there are no transmission blocks.
\label{theoremA0}
\end{theorem}
\begin{proof}
The Theorem follows from Corollary~\ref{theoremA2} and Lemma~\ref{theoremA3} (see Appendix A). From Corollary~\ref{theoremA2}, the individuals contained in a subsystem state larger than a pair appearing in $M_E$ belong to some transmission block. From Lemma~\ref{theoremA3}, any transmission block appears as a subsystem state in $M_E$.   
\end{proof}

\section{Hierarchies of approximate models}
\label{4}
The systems of equations in the previous section are exact, but limited in applicability because of the limited scope for dynamical partitioning in most networks. To suitably curtail the large number of equations, the networks need to have a structure which is roughly tree-like.

More typically, we want to trade off some exactness for models which are numerically tractable and provide a good, rather than exact description of the underlying dynamics. The pair-level SIR model (equation~\ref{pair_level}) is exact for tree networks but is also a reasonably good approximation for SIR dynamics on a wide range of networks. Higher-order models will typically be more accurate, but will have considerably greater computational cost. Here we formally define hierarchies of approximate models that can be applied to Markovian SIR dynamics on any network.

We define `pseudo-partitioning' according to different criteria. We define two hierarchies of models via what we term `cycle-partitioning' and `size-partitioning'. We then also consider a `hybrid-partitioning'  hierarchy utilising both methods. Although these pseudo-partitionings can be defined more generally, as in the case of dynamical partitioning itself, we shall restrict our attention here to dynamical partitioning with respect to single susceptible nodes.

Generalising from the case of dynamical partitioning, we define a function $f_p(X,Y,i)$ to specify some pseudo-partitioning of subsets $X$ and $Y$ with respect to node $i$ and enable a systematic curtailing of the number of subsystem states necessary for a solvable model. By analogy with equation~\ref{part_eq1}, we have:
\begin{eqnarray} \nonumber
\dot{\la \psi^A_W\ra}&\approx &\sum_{k \in W} \mathbbm{1}(A_k= S) \Bigg[ - \sum_{n \in W} T_{kn} \mathbbm{1}(A_n= I)  \la \psi^A_W \ra \Bigg] \\ \nonumber
&& + \sum_{k \in W} \mathbbm{1}(A_k= I) \Bigg[  \sum_{n \in W} T_{kn} \mathbbm{1}(A_n= I) \la h_k^S( \psi^A_W) \ra - \gamma_k \la \psi^A_W \ra \Bigg] \\ \nonumber
&&- \sum_{k \in W} \mathbbm{1}(A_k= S) \sum_{n \in V \setminus W}T_{kn} \Bigg[  \Big( 1-f_p(n,W \setminus k,k) \Big) \la h_k^S(\psi^A_W) I_n\ra \\ \nonumber
&& + f_p(n,W \setminus k,k)\frac{\la h^S_k(\psi^A_W)\ra \la S_k I_n\ra}{\la S_k\ra}  \Bigg] \\ \nonumber
&&+ \sum_{k \in W} \mathbbm{1}(A_k= I) \sum_{n \in V \setminus W}T_{kn} \Bigg[  \Big( 1-f_p(n,W \setminus k,k) \Big) \la h_k^S(\psi^A_W) I_n\ra \\ 
&& + f_p(n,W \setminus k,k)\frac{\la h^S_k(\psi^A_W)\ra \la S_k I_n\ra}{\la S_k\ra}  \Bigg] . 
\label{part_eq}
\end{eqnarray}
So, when $f_p(X,Y,i)=1$, we treat $i$ as if it is dynamically partitioning with respect to $X$ and $Y$ and so the right-hand-side of the rate equation does not generate larger subsystem states. The specific type of approximate model depends on how $f_p(X,Y,i)$ is defined and is formed by assuming equality between the left and right hand sides.

Note that equation~\ref{part_eq} defines a solvable model that is based on the closure in equation~\ref{eqclosure}. However, other closures such as the Kirkwood-closure $\la \psi^A_i\psi^B_j\psi^C_k\ra \approx \la \psi^A_i\psi^B_j\ra\la \psi^B_j\psi^C_k\ra\la \psi^C_k\psi^A_i\ra/(\la \psi^A_i\ra\la \psi^B_j\ra\la \psi^C_k\ra)$ fall outside of this scheme. It is, however, straightforward to define a solvable hierarchy of approximate models that incorporates the standard Kirkwood closure as a special case. 

Let us denote the adjacency matrix for the underlying graph by $U$ ($U_{ij}=\textrm{sgn} (T_{ij}+T_{ji})$ for all $i, j\in V$). Then, for the probability of subsystem $\psi_W$ being in state $A$, we can approximate:
\beq
\la \psi^A_W\ra\approx\frac{\prod\limits_{i,j\in W : j<i}\la \psi^{A_i}_i\psi^{A_j}_j\ra^{U_{ij}}}{\prod\limits_{i\in W}\la \psi^{A_i}_i\ra^{m_i-1}}
\label{K_approx}
\eeq
where $m_i=\sum_{j\in W}U_{ij}$ is the number of neighbours of node $i$ in $W$ in the underlying graph and is also the number of times that the state of node $i$ appears on the numerator. For a fully connected subsystem of three nodes, this is seen to reproduce the standard Kirkwood closure. Using this general idea but keeping the state $\psi_W^A$ intact as before, we can write an alternative to equation~\ref{eqclosure}. For $\psi_W^A$, if $A_k=S$, $k\in W$ and $n\in V\setminus W$ and we suppose there is an arc from $n$ to $k$, then we can approximate: 
\begin{eqnarray}
\la \psi^A_WI_n\ra & \approx & \frac{ \la \psi_W^A \ra \la S_k I_n \ra}{ \la S_k \ra} \prod_{j \in W \setminus k}\left[ \frac{ \la \psi_j^{A_j} I_n \ra}{\la \psi_j^{A_j} \ra \la I_n \ra} \right]^{U_{nj}}. 
\label{K_approx1}
\end{eqnarray}
Examples of the application of this approximation can be found in section~\ref{5.5}.

We use this approximation to motivate the following hierarchy:
\begin{eqnarray} \nonumber
\dot{\la \psi^A_W\ra}&\approx&\sum_{k \in W} \mathbbm{1}(A_k= S) \Bigg[ - \sum_{n \in W} T_{kn} \mathbbm{1}(A_n= I)  \la \psi^A_W \ra \Bigg] \\ \nonumber
&& + \sum_{k \in W} \mathbbm{1}(A_k= I) \Bigg[  \sum_{n \in W} T_{kn} \mathbbm{1}(A_n= I) \la h_k^S( \psi^A_W) \ra - \gamma_k \la \psi^A_W \ra \Bigg] \\ \nonumber
&&- \sum_{k \in W} \mathbbm{1}(A_k= S) \sum_{n \in V \setminus W}T_{kn} \Bigg[  \Big( 1-f_p(n,W \setminus k,k) \Big) \la h_k^S(\psi^A_W) I_n\ra \\ \nonumber
&& + f_p(n,W \setminus k,k)\frac{ \la  h_k^S(\psi_W^A) \ra \la S_k I_n \ra}{ \la S_k \ra} \prod_{j \in W \setminus k}\left( \frac{ \la \psi_j^{A_j} I_n \ra}{\la \psi_j^{A_j} \ra \la I_n \ra} \right)^{U_{nj}} 
 \Bigg] \\ \nonumber
&&+ \sum_{k \in W} \mathbbm{1}(A_k= I) \sum_{n \in V \setminus W}T_{kn} \Bigg[  \Big( 1-f_p(n,W \setminus k,k) \Big) \la h_k^S(\psi^A_W) I_n\ra \\ 
&& + f_p(n,W \setminus k,k) \frac{ \la h_k^S(\psi_W^A) \ra \la S_k I_n \ra}{ \la S_k \ra} \prod_{j \in W \setminus k}\left( \frac{ \la \psi_j^{A_j} I_n \ra}{\la \psi_j^{A_j} \ra \la I_n \ra} \right)^{U_{nj}}  \Bigg] .
\label{part_eq2}
\end{eqnarray}

Either equation~\ref{part_eq} or~\ref{part_eq2} can be used in conjunction with suitable definitions of $f_p(X,Y,i)$ to generate hierarchies of approximate models. We shall mostly use equation~\ref{part_eq} for explicit examples. However, for completeness, we shall briefly discuss equation~\ref{part_eq2} in section~\ref{5.5}.

It is worth noting that both of these closures are based around a single $IS$ arc being added each time. Other schemes with more complex closures should also be possible. For example, Theorem~\ref{general_partitioning_theorem} allows closures where we do not necessarily need to have only singlet states in the denominator (see Figure~\ref{general_exact_closures}b).

\subsection{Cycle-partitioning}
\label{4.1}
With reference to Figure~\ref{partitioning}, although node $i$ is not dynamically partitioning with respect to $W \setminus i$ and node $j$, we might observe that it is in some sense `approximately' dynamically partitioning because the path length between $j$ and $W$ is reasonably long when $i$ is deleted. It seems sensible to define a type of pseudo-partitioning according to this path length.
\begin{figure}
   \centerline{\includegraphics[width=.5\textwidth]{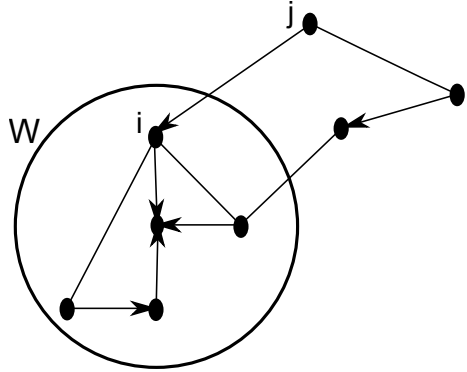}}
  \caption{An example of a node $i\in W$ which is not dynamically partitioning with respect to node $j$ and $W \setminus i$, but it is cycle-partitioning up to $x=2$.}
 \label{partitioning}
\end{figure}
\begin{definition}
The set of individuals that can reach at least one member of $X\subseteq V$, by traversing $a \in \mathbbm{N}$ arcs or less, is denoted $IN_a(X)$. Here and elsewhere, $\mathbbm{N}=\{0,1,2,...\}$.
\end{definition}

\begin{definition}
Node $i \in V$ is `cycle-partitioning' at order $x \in \mathbbm{N}$ with respect to disjoint and non-empty subsets $X, Y \subset V$, where $i\notin X\cup Y$, if and only if we have $f_{C(x)}(X,Y, i)=1$ where:
\beq
f_{C(x)}(X,Y, i)= \begin{cases}
1 & \mbox{if } \mbox{IN}_a(X) \cap \mbox{IN}_b(Y) = \emptyset  \quad \forall a,b: a+b=x \quad \mbox{ \rm{(}in} \quad D-i) \\
0 & \mbox{otherwise}
\end{cases}
\eeq
where $a,b \in \mathbbm{N}$.
\end{definition}

We make the following observations: i) If the network is undirected then $f_{C(x)}(X,Y, i)=0$ if and only if there is at least one path of length $x$ or less between some member of $X$ and some member of $Y$ when $i$ is deleted. ii) An individual who is dynamically partitioning with respect to two subsets is also cycle-partitioning at all orders with respect to those subsets. iii) In Figure~\ref{partitioning}, node $i$ is cycle-partitioning with respect to $W\setminus i$ and $j$ for $x=0$, $x=1$, and $x=2$, but not $x>2$. iv) Any individual $i \in V$ is always cycle-partitioning at order $x=0$ with respect to any other two subsets. 

Adapting Corollary~\ref{partitioning_theorem} such that cycle-partitioning individuals of order $x \in \{0,1,2, \ldots \}$ are `treated' as dynamically partitioning individuals, we substitute $f_{p}(n,W \setminus k, k)=f_{C(x)}(n,W \setminus k, k)$ into equation~\ref{part_eq}.

\begin{remark}
By applying this rate equation to every individual in the network for states $S$ and $I$ and then reapplying to every new subsystem state which emerges, we obtain a closed set of differential equations which form the $x$th model in a hierarchy of approximating models (note that the model corresponding to $x=0$ is the pair-level model given by equation~\ref{pair_level}). The associated set of induced subsystem states will be denoted by $M_{C(x)}$.
\end{remark}

\subsubsection{Examples}

We can consider cycle-partitioning for the network in Figure~\ref{partitioning_examples}b. If we cycle-partition at $x=1$, then the first two terms of equation~\ref{squarepairs} are closed at the level of pairs. Specifically, for the first term, node 2 is cycle-partitioning with respect to nodes 1 and 3. For the second term, node 1 is cycle-partitioning with respect to nodes 2 and 4. This gives:
\beq \nonumber
\dot{\la S_1I_2\ra}&\approx&\frac{\la S_1S_2 \ra \la S_2 I_3\ra}{\la S_2 \ra}- \frac{\la S_1I_2 \ra \la S_1 I_4\ra}{\la S_1 \ra}-\frac{\la S_1I_2\ra\la S_1I_5\ra}{\la S_1\ra}-\frac{\la S_1I_2\ra\la S_1I_6\ra}{\la S_1\ra} \\ \nonumber
&&-(1+g)\la S_1I_2\ra.
\eeq
Thus, triples within the square are no longer `kept intact', and so, within the square, the model closes at the level of pairs. However, triples made up of the members of the triangle are kept intact. For example, we have:
\beq\nonumber
\dot{\la S_5I_6\ra}=\la I_1S_5 S_6 \ra- \la I_1 S_5 I_6 \ra  - (1+g)\la S_5 I_6 \ra .
\eeq

Figure~\ref{cycle_partitioning} 
\begin{figure}
   \centerline{\includegraphics[width=1\textwidth]{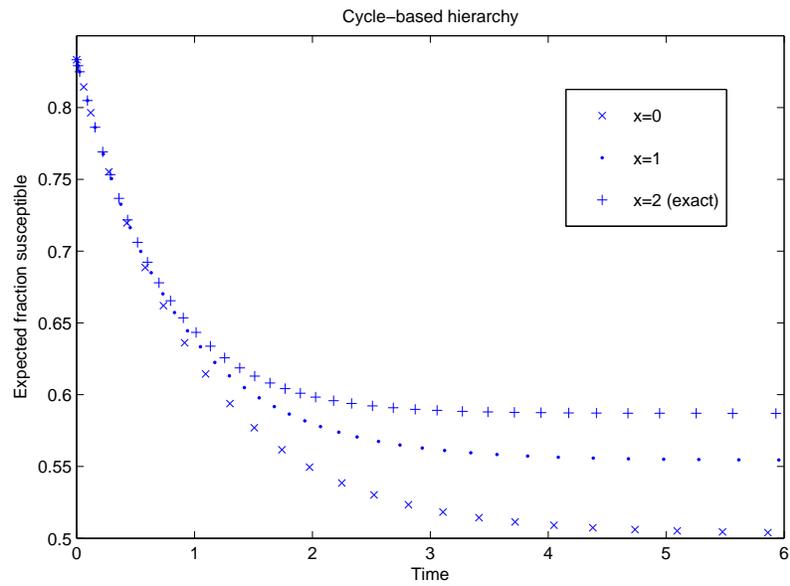}}
  \caption{Cycle-partitioning applied to the scenario in Figure~\ref{partitioning_examples}b with $x=0$ which corresponds to the pair-level model through $x=1$ and finally $x=2$ which is exact for this scenario. Here we assume that all individuals are susceptible at time $t=0$ with probability $5/6$ and infected otherwise (the states of individuals are initially statistically independent). We have assumed a transmission rate of unity across each link and a removal rate of unity.}
 \label{cycle_partitioning}
\end{figure}
shows this hierarchy of models. Here, the $x=0$ model is the pair-level model (equation~\ref{pair_level}). The $x=1$ model is an improvement since it picks up the triangle. The $x=2$ model picks up the square as well and is equivalent to the exact closure model (consistent with the master equation). 

If we apply cycle-partitioning to Figure~\ref{partitioning_examples}c instead, then the $x=0$ model is the pair-level model as always. The $x=1$ model is also the pair-level model and the $x=2$ model is equivalent to the exact closure model. Hence, cycle-partitioning does not necessarily lead to improved models as $x$ increases and it does not always lead to a reduction in system size with respect to the exact closure model. The results from the $x=0$ pair-level model and the exact model applied to Figure~\ref{partitioning_examples}c can be seen in the section on size-partitioning below (Figure~\ref{motif_partitioning}) and so are not reproduced here. 

An extreme example of the failure of cycle-partitioning to produce large hierarchies of approximate models is given by the triangular lattice shown in Figure~\ref{triangle_lattice}. 
\begin{figure}
  \centerline{\includegraphics[width=.6\textwidth]{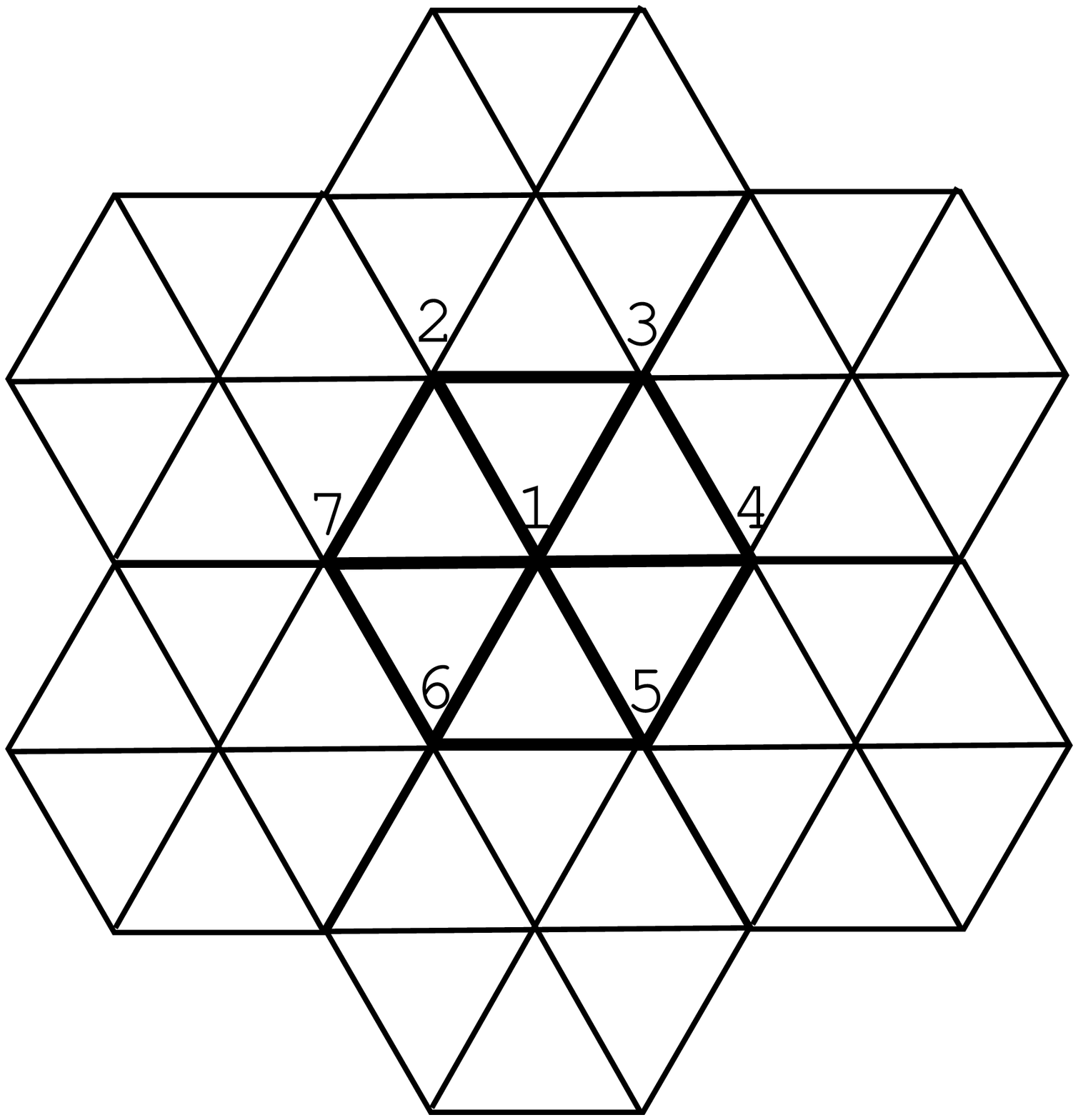}}
  \caption{A triangle lattice - an extreme example where cycle-partitioning at order greater than $x=0$ requires subsystem states which contain all individuals.}
 \label{triangle_lattice}
\end{figure}
Here, the $x=0$ model is the pair-level model. For $x=1$, consider the triple $A_3S_1C_4$ ($\forall A,C\in\{S,I\}$). Here we do not have cycle-partitioning since by deleting node 1, there is a path of length 1 between nodes 3 and 4. As we move to order 4 motifs, (e.g. adding a node to the above triple either by the edge (1,2) or the edge (3,2)), it is readily seen that there will always exist motifs which do not cycle-partition for $x=1$ at all orders. Hence for the triangle lattice, even for $x=1$ cycle-partitioning, we obtain a model with motif states at the size of the full network. Some cycle-partitioning does occur however, so the resulting model is not exact. For example, for the triple $A_2S_1C_4$, deleting node 1 means that the shortest path from 2 to 4 is via node 3 and is of length 2. So we have cycle-partitioning here. We also have it for states $A_7S_1C_4$. This state is also cycle-partitioning at $x=2$ (the path length from node 7 to node 4 after deletion of node 1 is 3) but we no longer cycle-partition $A_2S_1C_4$. Finally, at $x=3$, no cycle-partitioning occurs anywhere and we have an exact model containing subsystem states at the size of the system ($M_{C(3)}=M$).

In general, if the largest transmission block in a network has $n$ individuals, then any cycle-partitioning model of order $x\ge n-2$ is exact (see Theorem~\ref{theoremB1} in Appendix B). This is illustrated by the network in Figure~\ref{partitioning_examples}b where the largest transmission block is of size $n=4$ and the $x=2$ cycle-partitioning model is exact (Figure~\ref{cycle_partitioning}). This is also the case for the graph in Figure~\ref{partitioning_examples}c where $n=5$ and the $x=3$ model is exact (the $x=2$ model is also happens to be exact here as well). Another general result is that if the smallest directed sub-block consists of $n$ individuals, then the cycle-partitioning models of order $x<n-2$ are all equivalent to the pair-level (x=0) models (see Theorem~\ref{theoremB2} in Appendix B). This is illustrated by the graph in Figure~\ref{partitioning_examples}c where the smallest directed sub-block is $n=4$, and we found that the $x=1$ cycle-partitioning model is the same as the pair-level model.

\subsection{Size-partitioning}
\label{4.2}
The issues arising in some networks such as Figure~\ref{partitioning_examples}c, where even cycle-partitioning at $x=2$ requires subsystem states containing all individuals, and the extreme example of the triangular lattice, motivate an alternative pseudo-partitioning approach whereby the sizes of subsystem states are more directly constrained. 
\begin{definition}
\beq
f_{S(x)}(X)= \begin{cases}
1 & \mbox{if } |X|=x+1 \\
0 & \mbox{otherwise}
\end{cases}
\eeq
where $X \subseteq V$ and $x \in \mathbbm{N}$.
\end{definition}
Here we make the substitution $f_{p}(n,W \setminus k, k)=f_{S(x)}(W \setminus k)$ into equation~\ref{part_eq}.

\begin{remark}
As with previous pseudo-partitioning, a complete approximate model arises from the equations for the individual-level states and then repeatedly applying this equation to each subsystem state that emerges. As with cycle-partitioning, the $x=0$ size-partitioning model corresponds to the pair-level model. 
\end{remark}
\subsubsection{Examples}
As an example, consider the $x=1$ size-partitioning model for Figure~\ref{partitioning_examples}c, where the cycle-partitioning hierarchy was redundant.
Equation~\ref{eqp2} now becomes:
\begin{eqnarray} \nonumber
\dot{\la S_1S_2I_3\ra}&\approx&\frac{\la S_1S_2S_3\ra\la S_3I_4\ra}{\la S_3\ra}+\frac{\la S_1S_2S_3\ra\la S_3I_5\ra}{\la S_3\ra} -\frac{\la S_1S_2I_3\ra\la S_1I_4\ra}{\la S_1\ra}\\ \nonumber
&& \frac{\la S_1S_2I_3\ra\la S_1I_5\ra}{\la S_1\ra}-(1+g)\la S_1S_2I_3\ra .
\end{eqnarray}
For $x=2$ size-partitioning, equation~\ref{eqp2} is left untouched since the exact rate equation for a subsystem state of size 3 does not involve subsystem states larger than 4. However, equation~\ref{eqp3} becomes:
\beq
\dot{\la S_1S_2S_3I_4\ra}&\approx&-\frac{\la S_1S_2S_3I_4\ra\la S_1I_5\ra}{\la S_1\ra}-\frac{\la S_1S_2S_3I_4\ra\la S_3I_5\ra}{\la S_3\ra}-(2+g)\la S_1S_2S_3I_4\ra .
\nonumber
\eeq
 In this way, we obtain three different approximate models: $x=0$, $x=1$ and $x=2$. For $x>2$, the model is exact.
Figure~\ref{motif_partitioning} shows results from the application of each of these three approximate models and the exact $x=3$ model to SIR dynamics on the network depicted in Figure~\ref{partitioning_examples}c. An interesting feature that should be noted for the $x=2$ model is that it very slightly underestimates the rate of spread of the epidemic. Typically, experience shows that the closure of these equations leads to over-estimation of the rate of spread, but this provides a counter example.

\begin{figure}
   \centerline{\includegraphics[width=1\textwidth]{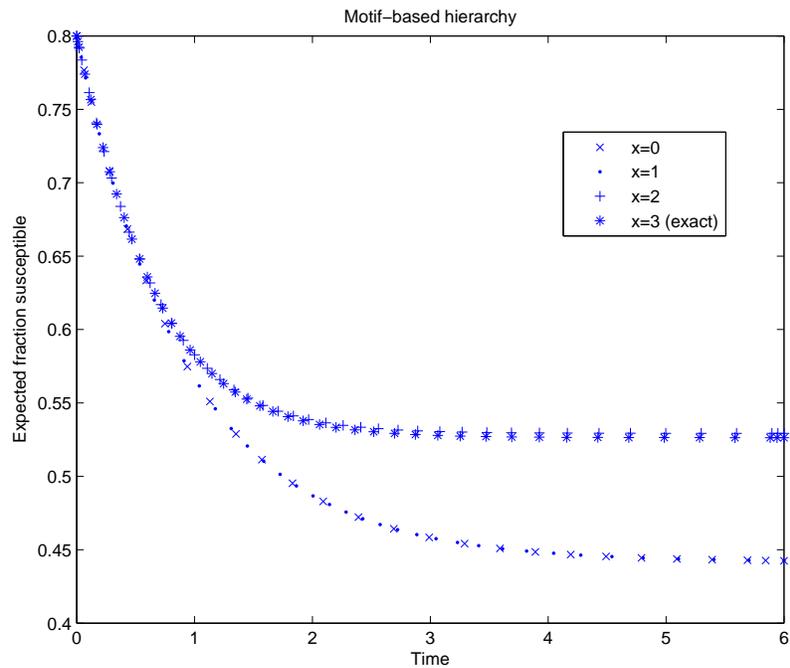}}
  \caption{Size-partitioning applied to the scenario in Figure~\ref{partitioning_examples}c with $x=0$ which corresponds to the pair-level model through $x=1$, $x=2$, and finally $x=3$ which is exact for this scenario. An individual is assumed to be initially susceptible with probability $4/5$ and infected otherwise (the states of individuals are initially statistically independent). We have assumed a transmission rate of unity across each link and a removal rate of unity.}
 \label{motif_partitioning}
\end{figure}
While size-partitioning will generate a large hierarchy of approximate models where cycle partitioning fails to do so (such as for the triangular lattice), it has problems of its own. Specifically, we see from Figure~\ref{motif_partitioning} that since the smallest cycle size in Figure~\ref{partitioning_examples}c is 4, the $x=1$ size-partitioning model is almost identical to the $x=0$ pair-level model. The $x=3$ and $x=2$ models are also almost identical. Hence, the extra computation in evaluating at $x=1$ and $x=3$ is wasteful. In this sense, cycle-partitioning has an advantage by only picking up complete cycles in the network. 

An additional problem with size-partitioning is that it ignores genuine dynamical partitioning. For example, for Figure~\ref{partitioning_examples}b, we would require motif sizes of 6 ($x=4$) to describe this exactly within the size-partitioning scheme. However, if we permit genuine dynamical partitioning, we only need motif sizes of less than or equal to 4. This issue is readily resolved by considering the modified scheme:
\beq
f_{E,S(x)}(X,Y,i)= \begin{cases}
1 & \mbox{if } f_E(X,Y,i)=1 \mbox{ or } f_{S(x)}(Y)= 1 \\
0 & \mbox{otherwise}
\end{cases}
\eeq
which incorporates genuine dynamical partitioning into size-partitioning. With this rule, in Figure~\ref{partitioning_examples}b, the genuine dynamical partitioning around node 1 is utilised wherever possible.

\subsection{Hybrid-partitioning}
Both cycle-partitioning and size-partitioning have their merits. Size-partitioning avoids unnecessarily large motif states where cycle-partitioning cannot be effectively implemented beyond an early stage, such as in the triangle lattice. On the other hand, cycle-partitioning picks out cycles in the network and closes at the pair level unless complete cycles can be incorporated, avoiding wasteful computation with minimal gain in accuracy.

We can construct a hybrid-partitioning scheme that captures the benefits of both cycle-partitioning and size-partitioning while avoiding the problems of both. We define this hybrid-partitioning as:

\begin{definition}
\beq
f_{C(x)S(x)}(X,Y,i)= \begin{cases}
1 & \mbox{if } f_{C(x)}(X,Y,i)=1 \mbox{ or } f_{S(x)}(Y)= 1 \\
0 & \mbox{otherwise}
\end{cases}.
\eeq
\end{definition}
This leads to a hierarchy of models defined by substituting $f_{p}(n,W \setminus k, k)=f_{C(x)S(x)}(n,W \setminus k, k)$ into equation~\ref{part_eq}.
This also has the pair-level model for $x=0$. We also note that alternative hierarchies could be designed with different values of $x$ for the size-partitioning and the cycle-partitioning parts.

This closure benefits from the advantages of both cycle-partitioning and size-partitioning. Firstly, if there are only large cycles, the hierarchy is closed at a low order by cycle-partitioning. This is desirable since, as illustrated in Figure~\ref{motif_partitioning}, continuing on generates little benefit unless we are able to continue to the size of the smallest cycle. However, if the system is not amenable to cycle-partitioning, as in the triangular lattice, then size-partitioning is required. A network illustrating the benefits of this is shown in Figure~\ref{hybrid_partitioning}. 
\begin{figure}
   \centerline{\includegraphics[width=.5\textwidth]{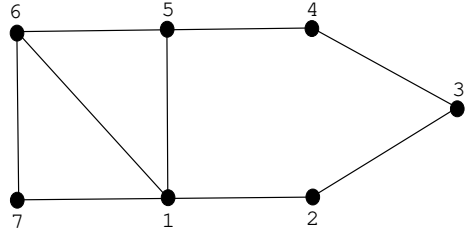}}
  \caption{A graph that illustrates the benefits of hybrid-partitioning. Expanding from node 1 using $x=1$, we utilise both cycle-partitioning and size-partitioning capturing the advantages of both.}
 \label{hybrid_partitioning}
\end{figure}
For hybrid-partitioning with $x=1$, let us start with the probability that node 1 is infectious:
\beq
\dot{\la I_1\ra}=\la S_1I_2\ra+\la S_1I_5\ra+\la S_1I_6\ra+\la S_1I_7\ra-g\la I_1\ra.
\label{I_1eq}
\eeq
For the first of these terms on the right-hand-side, the corresponding approximate differential equation is:
\beq \nonumber
\dot{\la S_1I_2\ra}&\approx&\frac{\la S_1S_2\ra\la S_2I_3\ra}{\la S_2\ra}-\frac{\la S_1I_2\ra\la S_1I_5\ra}{\la S_1\ra}-\frac{\la S_1I_2\ra\la S_1I_6\ra}{\la S_1\ra} \\ \nonumber
&&-\frac{\la S_1I_2\ra\la S_1I_7\ra}{\la S_1\ra}-(1+g)\la S_1I_2\ra
\eeq
where we have employed $x=1$ cycle-partitioning. For the term $\la S_1I_5\ra$ in equation~\ref{I_1eq} we obtain:
\beq\nonumber
\dot{\la S_1I_5\ra}&\approx& \frac{\la S_1S_5\ra\la I_4S_5\ra}{\la S_5\ra}+\la S_1S_5I_6\ra-\frac{\la S_1I_5\ra\la S_1I_2\ra}{\la S_1\ra} \\ \nonumber
&& -\la S_1I_5I_6\ra- \frac{\la S_1I_5\ra\la S_1 I_7 \ra}{\la S_1\ra}
\eeq
where, again, $x=1$ cycle-partitioning has been implemented where possible. For the second term in this expression, we have:
\beq \nonumber
\dot{\la S_1S_5I_6\ra}&\approx&\frac{\la S_1S_5S_6\ra\la S_6I_7\ra}{\la S_6\ra}-\frac{\la S_1S_5I_6\ra\la S_1I_7\ra}{\la S_1\ra} \\ \nonumber
&&-\frac{\la S_1S_5I_6\ra\la I_4S_5\ra}{\la S_5\ra}-\frac{\la S_1S_5I_6\ra\la S_1I_2\ra}{\la S_1\ra} - (2+g)\la S_1S_5I_6\ra.
\eeq
Here, the closures on the first line are via $x=1$ size-partitioning, whereas the closures on the second line are via meeting the criteria for both $x=1$ size-partitioning and $x=1$ cycle-partitioning.

So, this hybrid-partitioning obtains the best of both methodologies. Cycle-partitioning avoids unnecessarily including extra terms in the large cycle 1-2-3-4-5-1 which we have seen (Figure~\ref{motif_partitioning}) generates minimal extra accuracy. Size-partitioning forces partitioning where the motif sizes get beyond a specified level, here constraining the maximum motif size to be 3.

\subsection{Alternative closure}
\label{5.5}
Before leaving this section, we include a brief aside on using the alternative closure defined in equation~\ref{K_approx1}. In this case, we can still apply the cycle, size and hybrid methods, but we use equation~\ref{part_eq2} in place of equation~\ref{part_eq}. Two examples of applying this are illustrated in Figure~\ref{kirkwood_closure}. 
\begin{figure}
   \centerline{\includegraphics[width=.8\textwidth]{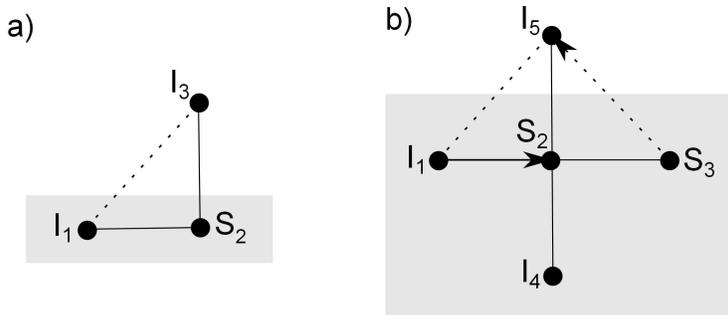}}
  \caption{Two simple examples of applying the alternative closure rule as encoded by equation~\ref{part_eq2}. The shaded region specifies the initial subsystem state, and there is a new $IS$ link towards it in accordance with the way in which the induced state spaces are built. The dashed lines represent additional links between the new node and the original subsystem (these would be ignored by the closure rule in equation~\ref{part_eq}.)}
 \label{kirkwood_closure}
\end{figure}
Here the shaded regions represent the existing subsystem states and the solid lines coming out of these regions represents the new infectious node being added on. The dashed lines represent other links between the new infectious nodes and the original subsystems. Supposing that the criteria for pseudo-partitioning is met at this stage (i.e. the relevant $f_p(.)=1$), for Figure~\ref{kirkwood_closure}a we obtain
\beq
\la I_1S_2I_3\ra\approx\frac{\la I_1S_2\ra\la S_2I_3\ra\la I_1I_3\ra}{\la I_1\ra\la S_2\ra\la I_3\ra},
\nonumber
\eeq
and for Figure~\ref{kirkwood_closure}b, we obtain
\beq
\la I_1S_2S_3I_4I_5\ra\approx\frac{\la I_1S_2S_3I_4\ra\la I_1I_5\ra\la S_2I_5\ra\la S_3I_5\ra}{\la I_1\ra\la S_2\ra\la S_3\ra\la I_5\ra^2}.
\nonumber
\eeq
We note that for cycle-partitioning with $x>0$, both closure methods become equivalent (equation~\ref{part_eq2} reduces to equation~\ref{part_eq}) since the types of additional links drawn in Figure~\ref{kirkwood_closure} could not be present.

Notice that when the closure of triples always occurs (e.g. $x=0$ cycle-partitioning or $x=0$ size-partitioning), the variant of the pair-level models introduced in \cite{Sharkey08} and \cite{Sharkey11} is obtained under this closure. This variant is expected to be able to handle clustered networks more accurately than the variant considered in \cite{Sharkey13} and \cite{Kiss} that follows from equation~\ref{part_eq}.

\section{Discussion}
Recently it has been possible to establish exact and practicable representations of stochastic epidemic dynamics on finite tree networks \cite{Sharkey13} using closure methodologies evaluated at the level of individuals \cite{Sharkey08,Sharkey11}. Message-passing also gives exact representations on trees \cite{KarrerNewman} and this can be shown, under some circumstances with Poisson transmission processes, to be equivalent to moment closure models \cite{Wilkinson}. Under suitable and very restrictive homogeneity assumptions, population-level versions of these closed models (e.g. \cite{Keeling99}) can also be exactly derived  on idealised graphs with homogeneous initial conditions \cite{Sharkey08}. 

Within the individual-level closure construction, it is possible to go beyond trees and obtain exact representations of epidemic dynamics on some networks with cycles using the idea of dynamical partitioning on the graph \cite{Kiss}. Here we defined dynamical partitioning on arbitrary networks and also observed that it applies to both Markovian and non-Markovian SIR dynamics. In the Markovian case with Poisson transmission and removal processes, we can use dynamical partitioning to define exact SIR moment closure models. The extent to which these models are computationally viable depends primarily on the underlying structure of the network.

More specifically, starting from the probabilities of the states of individual nodes in a given network, we uniquely defined the full set of exact induced moment equations by automatically implementing dynamical partitioning where applicable. We also defined transmission blocks as a natural decomposition of a network for the closure of SIR models. Transmission blocks represent a possible extension of the block concept in graph theory into directed networks. Using this concept, we proved a theorem stating that the size of the largest subsystem state appearing in the set of moment equations is equal to the size of the largest transmission block.

We also investigated hierarchies of approximate moment closure models. In the epidemic literature, it is normally the case that moment closure models are constructed at the level of pairs, or occasionally for triples or quads \cite{Matsuda, Bauch,House}. This is often accompanied with an assertion that higher order models exist. However, to our knowledge, these higher order epidemic models have never been defined explicitly. This is understandable since these models rapidly become too complex to be of real practical relevance, but it does leave open the theoretical question of how these models can be defined \cite{Sharkey11}. To address this, we introduced `pseudo-partitioning' to construct complete hierarchies of approximate closed models that are well-defined at all orders. In fact, we defined several hierarchies of closed models; one in terms of motif size, one in terms of the size of cycles in the network, and a hybrid method taking the best of both of the previous methods. Undoubtedly other hierarchies can be defined as well. In addition, we investigated two mechanisms of closure - one based on exact dynamical partitioning and the other which is more related to the Kirkwood closure. 

The closure based directly around dynamical partitioning has the variant of the closure model considered by \cite{Sharkey13} as its zeroth order variant (for all of the size, cycle and hybrid approaches). The hierarchies based around the alternative closure all have the model introduced in \cite{Sharkey08} and \cite{Sharkey11} as their zeroth order variant (this is designed to handle networks with clustering in a more effective way).  We also observed that the conditions for cycle-partitioning at orders greater than zero mean that both methods of closure become equivalent.

The hierarchies of models generated some interesting observations concerning the convergence to exactness with order. For example, for size-partitioning, the models converge to the exact solution with increasing order, but this convergence is not always monotonic (see Figure~\ref{motif_partitioning}). It is typical for moment closure models of SIR epidemics to over-exaggerate the spread of an epidemic, but here we observed a counter example (see also \cite{Sharkey11} where this is discussed as a possibility). An unanswered question is whether the approximate models always increase in accuracy as the order of the hierarchy increases. Intuitively we would expect that they do, and this is validated by the examples so far investigated.

\section*{Acknowledgements}
\noindent RRW acknowledges support from EPSRC (DTA studentship). We thank two anonymous reviewers for helpful suggestions. 

\appendix

\section{Proof of the underpinning results for Theorem~\ref{theoremA0}}
Theorem~\ref{theoremA0} follows from Corollary~\ref{theoremA2} and Lemma~\ref{theoremA3} below.
\begin{definition}
A set $W_n \subset V$ of size $|W_n|=n$ can be `generated' from a set $W_m \subset V$ of size $|W_m|=m$ where $2\le m<n$ if and only if a sequence of sets $W_m, \ldots, W_i, \ldots , W_n$ exist where $W_{i+1}=W_i \cup \{k\}$, where $k$ is a single node in $V \setminus W_i$, and there exists an arc from $k$ towards some individual $j \in W_i$ which is not dynamically partitioning relative to $k$ and $W_i \setminus \{j\}$.
\label{defnA1}
\end{definition}

\begin{remark}
The above definition is constructed such that there exists a subsystem state $A:W \to \{S,I \}$ belonging to $ M_E$, where $|W|>2$, if and only if $W$ can be generated from some connected pair. This follows from the definition of $M_E$ via equation~\ref{part_eq1}. 
\end{remark}
\begin{lemma}

If a set $W$ can be generated from some connected pair, then there exists $X\supseteq W$ such that $D[X]$ is a directed sub-block. There also exists some node $i\in W$ that is reachable from all other nodes in both $D[W]$ and $D[X]$.
\label{theoremA5}

\end{lemma}

\begin{proof} 
The proof follows by induction. Lemma~\ref{propA1} proves the statement for the case $|W|=3$ while Lemma~\ref{propA2} establishes the inductive step.
\end{proof}

\begin{cor}
If $A: W \to \{ S,I  \}$ is a subsystem state belonging to  $M_E$, where $|W| > 2$, then there exists $X \supseteq W$ such that $D[X]$ is a directed sub-block. 
\label{theoremA2}
\end{cor} 
\begin{proof}
This follows directly from Lemma~\ref{theoremA5} and Definition~\ref{defnA1}.
\end{proof}

\begin{lemma}
If a set $W$ where $|W| =3$ can be generated from some connected pair, then there exists $X \supseteq W$ such that $D[X]$ is a directed sub-block, and some $i \in W$ is reachable from all others in both $D[W]$ and $D[X]$. 
\label{propA1}
\end{lemma} 

\begin{proof}
We focus only on directed links since directed sub-blocks cannot be destroyed by making links undirected. With reference to Figure~\ref{drawing}, if a set $W_3=\{i,j,k\}$ can be generated from the pair $W_2=\{i,j\}$, with $j$ connected towards $i$, then there is a link from $k$ to either $i$ or $j$. Further, from the definition of dynamical partitioning and the generating rule, there are two possibilities: 1) there exists two vertex disjoint paths $P_1,P_2$ from some individual (which could be $k$) to both members of $W_2$, and where $k$ is the penultimate individual in one of these paths (see Figure~\ref{drawing}a and c), or 2) there exists a path $P_3$ from one member of $W_2$ to the other, and $k$ is the penultimate individual in this path (see Figure~\ref{drawing}b and d). Note that in all cases depicted in Figure~\ref{drawing}, $W_3$ is a subset of some directed sub-block in which $i$ is reachable from all others (and $i$ is reachable from all others in $D[W_3]$). 

\begin{figure}
   \centerline{\includegraphics[width=0.8\textwidth]{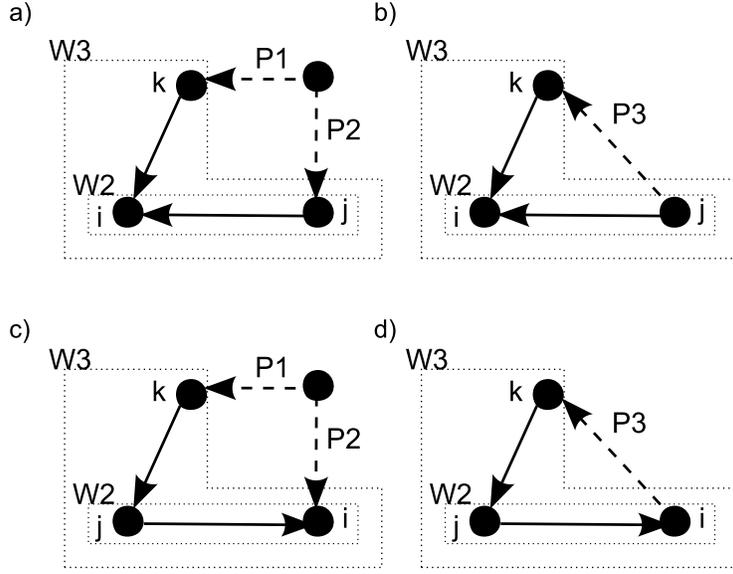}}
  \caption{Demonstration for Lemma~\ref{propA1}: `ways' in which a set $W_3=\{i,j,k \}$ can be generated from the pair $\{ i,j \}$, where $j$ is connected towards $i$. Note that $W_3$ is always a subset of some directed sub-block, and $i$ is reachable from all others in both $D[W_3]$ and the directed sub-block. The dashed arrows represent paths which may consist of any number of vertices.
}
 \label{drawing}
\end{figure}

\end{proof}

\begin{figure}
   \centerline{\includegraphics[width=0.8\textwidth]{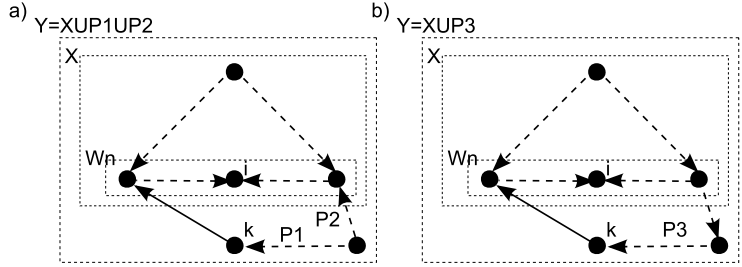}}
  \caption{Demonstration for Lemma~\ref{propA2}. Here, the single node in $X\setminus W$ is illustrative of the nodes in this set which must be connected by at least one path to node $i$, and where the underlying graph $G[X]$ is biconnected. We have placed node $k$ outside of $X$, but $k\in X\setminus W$ is also permitted. a) shows $k$ belonging to one of two vertex disjoint paths from some node to $W$ and b) shows $k$ as the penultimate individual in a path from a node in $W$ to a different node in $W$. In either case, $W_n \cup \{ k \}$ is seen to always be a subset of some $Y \supseteq X$ where $D[Y]$ is a directed sub-block in which $i$ is reachable from all others (and $i$ is reachable from all others in $D[W_n \cup \{ k \}]$).  }
 \label{figprop2}
\end{figure}

\begin{lemma}
If the statement made in Lemma~\ref{theoremA5} is true for the case where $|W|=n$, then it is also true when $|W|=n+1$.
\label{propA2}
\end{lemma}
\begin{proof}
Firstly, note that $W_{n+1}$, where $|W_{n+1}|=n+1$, can be generated from some connected pair if and only if it can be generated from some set $W_n$, where $|W_n|=n$, which can itself be generated from some connected pair. Now suppose that Lemma~\ref{theoremA5} is true for the case where $|W|=n$, and let $W_n$ be a set of size $n$ that can be generated from some connected pair. Then we have a set $X\supseteq W_n$ such that $D[X]$ is a directed sub-block where, without loss of generality, $i \in W_n \subseteq X$ is reachable from all others in both $D[W_n]$ and $D[X]$. With reference to Figure~\ref{figprop2}, and again focusing only on directed links, if a set $W_{n+1}= W_n \cup \{ k \}$ ($k \notin W_n$) can be generated from $W_n$, then either there exist two vertex disjoint paths $P_1,P_2$ from some individual to two different members of $W_n$ and $k$ is the penultimate individual in one of these paths (Figure~\ref{figprop2}a), or there exists a path $P_3$ from one member of $W_n$ to a different member of $W_n$ and $k$ is the penultimate individual in this path (Figure~\ref{figprop2}b). This follows from the generating rule and the definition of dynamical partitioning. Note that if $P_1, P_2$ exist then $D[X \cup P_1 \cup P_2]$ is a directed sub-block in which $i$ is reachable from all others (and $i$ is reachable from all others in $D[W_{n+1}]$). Similarly, if $P_3$ exists then $D[X \cup P_3]$ is a directed sub-block in which $i$ is reachable from all others (and $i$ is reachable from all others in $D[W_{n+1}]$).
\end{proof}

\begin{lemma}
If there exists $X \subset V$ such that $D[X]$ is a directed sub-block, then there exists a subsystem state $A:X \to \{ S,I \}$ belonging to $M_E$.
\label{theoremA3}
\end{lemma}

\begin{proof} 
If $D[X]$ is a directed sub-block in which $i \in X$ is reachable from all others, then there exists at least one arc $(j,i)$ in $D[X]$. The lemma then follows from lemma~\ref{propA3} below which proves that $X$ can be generated from $\{i,j\}$. 
\end{proof}
\begin{lemma}
Let $D[X]$ be a directed sub-block and let $i\in W \subset X$, where $|W| \ge 2$, be reachable from all others in both $D[W]$ and $D[X]$. In this case, some set $W \cup \{k\}$, where $k \in X \setminus W$, can be generated from $W$, and $i$ is reachable from all others in $D[W \cup \{ k \} ]$.
\label{propA3}

\end{lemma}
\begin{proof}
From Figure~\ref{figprop2} but with $k\in X$ we note that some set $W \cup \{k \}$, where $k \in X \setminus W$, can be generated from $W$ if and only if there exist two vertex disjoint paths $P_1,P_2$ from some individual to two different members of $W$ and where $k$ is the penultimate individual in one of these paths, or there exists a path $P_3$ from one member of $W$ to a different member of $W$ and $k$ is the penultimate individual in this path. Our proof is by contradiction. We shall assume that neither of these scenarios hold and show that this contradicts the assumption that $D[X]$ is a directed sub-block. 


 Every individual in $X \setminus W$ is at the start of a path to $i$ in $D[X]$. Figure~\ref{figprop3} shows the different ways in which an individual $k \in X \setminus W $ may be connected to an individual of $W$ in D[X]. Firstly, the underlying graph in Figure~\ref{figprop3}a is not biconnected so here $D[X]$ is not a directed sub-block. Secondly, Figures~\ref{figprop3}b and~c correspond to the existence of path $P3$ and the existence of paths $P1,P2$ respectively and hence $W \cup \{k \}$ is generated. Finally, Figure~\ref{figprop3}d has an individual from which $W$ cannot be reached and so $D[X]$ is not a directed sub-block. Other more complicated variants of this path will also contain such individuals from which $W$ cannot be reached. Hence, if paths $P1$ and $P2$ do not exist, and path $P3$ does not exist, then $D[X]$ is not a directed sub-block.

\begin{figure}
   \centerline{\includegraphics[width=.9\textwidth]{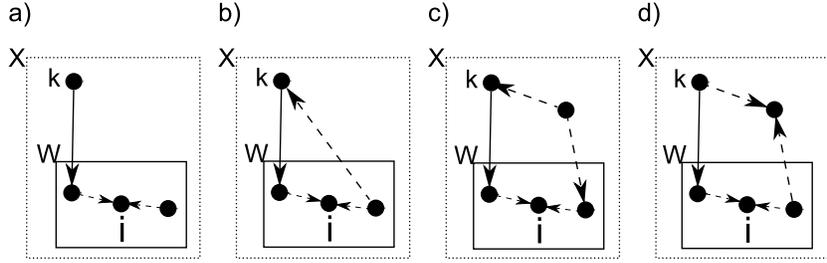}}
  \caption{Demonstration for Lemma~\ref{propA3}: shows the ways in which $k\in X \setminus W$ can be connected to $W$. We have cases a) the underlying graph of $D[X]$ is not biconnected, b) Existence of path $P3$, c) Existence of paths $P1$ and $P2$ and d) Existence of a node from which $W$ cannot be reached. Cases a) and d) are not directed sub-blocks so the existence of paths $P1$ and $P2$, or path $P3$, is established.}
 \label{figprop3}
\end{figure}
\end{proof}

\section{Proof of general results on cycle-partitioning}
The main results of this appendix are stated as Theorem~\ref{theoremB1} and Theorem~\ref{theoremB2}.
\begin{lemma}
\label{prop2.1}
Any induced subsystem state $A:W \to \{S,I \}$ belonging to $M_{C(x)}$ consists of a set of individuals $W \subset V$ where there is at least one individual reachable from all others in $D[W]$.
\end{lemma}
\begin{proof}
Follows from the way in which $M_{C(x)}$ is constructed via equation \ref{part_eq} (or equation \ref{part_eq2}).
\end{proof}
\begin{theorem}
If the largest transmission block in a network consists of $n$ individuals, then any cycle-partitioning model of order $x \ge n-2$ is exact. 
\label{theoremB1}
\end{theorem}

\begin{proof}
For any $W \subset V$ where at least one individual is reachable from all others in $D[W]$, if any $i \in W$ is cycle-partitioning at order $x \ge n-2$ with respect to some $j \notin W$ and $W \setminus i$, where $(j,i)$ is an arc, then $i$ is also dynamically partitioning with respect to $j$ and $W \setminus i$. This follows because if $i$ is not dynamically partitioning, but is cycle-partitioning at order $x>n-2$, then this implies the existence of a directed sub-block containing $j$, $i$ and at least one other member of $W$, and which consists of more than $n$ individuals. Therefore, by Lemma~\ref{prop2.1}, $M_{C(x)}$ only utilises genuine dynamical partitioning and we have $M_{C(x)}=M_E$ for $x \ge n-2$.  
 
\end{proof}

\begin{theorem}
If the smallest directed sub-block consists of $n$ individuals, then all cycle-partitioning models of order $x<n-2$ are equivalent to the pair-level models.
\label{theoremB2}
\end{theorem}
\begin{proof}

For any connected pair $W \subset V$ ($|W|=2$), if $i \in W$ is not cycle-partitioning at order $x < n-2$ with respect to $j \notin W$ and $W \setminus i$, where $(j,i)$ is an arc, then there exists a directed sub-block containing $W \cup j$, and which consists of less than $n$ individuals. Therefore, no such $j$ can exist. From the way in which $M_{C(x)}$ is constructed, this means that no subsystem states larger than connected pairs emerge and we have the pair-level model, i.e. $M_{C(x)}=M_{C(0)}$ for $x < n-2$.

\end{proof}

\begin{remark}
Together, Theorems~\ref{theoremB1} and~\ref{theoremB2} imply that the difference in size between the largest directed sub-block (or largest transmission block) and smallest directed sub-block gives an upper bound on the number of distinct models that the cycle-partitioning approach can provide. If all directed sub-blocks are the same size then no models that are distinct from the pair-level model and the exact dynamical partitioning model emerge. However, even when this difference is large the number of distinct models may sometimes be small, as was shown to be the case for the triangle lattice (where the difference is $N-3$). 
\end{remark}

\end{document}